%% file: main.tex
\documentclass{article}
\usepackage[utf8]{inputenc}
\usepackage[noend]{algpseudocode}
\usepackage{fancyvrb, algorithm, comment, amsthm, amsfonts}
\usepackage{amsmath, amssymb, amsthm}
\usepackage{mathtools, threeparttable, booktabs, caption, makecell,comment}
\usepackage{float}
\usepackage{multicol}
\usepackage[normalem]{ulem}
\usepackage{thmtools}
\usepackage{thm-restate}
\usepackage{cleveref}
\usepackage{lineno}
\usepackage{caption}
\usepackage{tabularx}
\title{Reliable Broadcast in Practical Networks: Algorithm and Evaluation}
\author{Yingjian Wu, Haochen Pan, Saptaparni Kumar, Lewis Tseng }
\date{wuit@bc.edu,  haochen.pan@bc.edu,  saptaparni.kumar@bc.edu,  lewis.tseng@bc.edu\\Boston College, MA, USA}

\DeclarePairedDelimiter\ceil{\lceil}{\rceil}
\DeclarePairedDelimiter\floor{\lfloor}{\rfloor}

\algnewcommand{\algorithmicand}{\textbf{ and }}
\algnewcommand{\algorithmicor}{\textbf{ or }}
\algnewcommand{\OR}{\algorithmicor}
\algnewcommand{\AND}{\algorithmicand}
\algnewcommand{\OnlyIf}[2]{\State #1\ \algorithmicif\ #2}
\algnewcommand{\var}{\texttt}

\newcommand{\msg}{\text{MSG}}
\newcommand{\echo}{\text{ECHO}}
\newcommand{\req}{\text{REQ}}
\newcommand{\fwd}{\text{FWD}}
\newcommand{\acc}{\text{ACC}}
\newcommand{\hashTag}{\text{HashTag}}
\newcommand{\codingTag}{\text{Co}}
\newcommand{\hash}{\mathbb{H}}
\newcommand{\mset}{\texttt{MsgSet}}
\newcommand{\cset}{\texttt{CodeSet}}
\newcommand{\hset}{\texttt{HashSet}}
\newcommand{\ctr}{\texttt{Counter}}
\newcommand{\true}{\texttt{True}}
\newcommand{\false}{\texttt{False}}
\newcommand{\oks}{\texttt{OK-Set}}
\newcommand{\nts}{\texttt{NT-Stars}}
\newcommand{\GG}{\mathbb{G}}

\theoremstyle{definition}

\newtheorem{theorem}{Theorem}

\newtheorem{lemma}[theorem]{Lemma}
\newtheorem{claim}[theorem]{Claim}
\newtheorem{property}[theorem]{Property}
\newtheorem{observation}[theorem]{Observation}

\newcommand{\sapta}[1]{{\color{red} **#1**}}

\newcommand{\lewis}[1]{{\color{blue} #1}}
\newcommand{\roger}[1]{{\color{cyan} #1}}
\newcommand{\commentOut}[1]{}

\begin{document}

\maketitle
\input{abstract}
\input{intro}
\input{model}   
\input{crash1}

\input{byz}
\input{EC-byz}

\input{Impossibility}

\input{eval}

\input{related}

\bibliographystyle{abbrv}
\bibliography{cite}

\end{document}

%% file: abstract.tex
\begin{abstract}


Reliable broadcast is an important primitive to ensure that a source node can reliably disseminate a message to all the non-faulty nodes in an asynchronous and failure-prone networked system. 
Byzantine Reliable Broadcast protocols were first proposed by Bracha in 1987, and have been widely used in fault-tolerant systems and protocols. Several recent protocols have improved the round and bit complexity of these algorithms.
    
Motivated by the constraints in practical networks, we revisit the problem.
In particular, we use \textit{cryptographic hash functions} and \textit{erasure coding} to reduce communication and computation complexity and simplify the protocol design. We also identify the fundamental trade-offs of Byzantine Reliable Broadcast protocols with respect to resilience (number of nodes), local computation, round complexity, and bit complexity.


Finally, we also design and implement a general testing framework for similar communication protocols. We evaluate our protocols using our framework. The results demonstrate that our protocols have superior performance in practical networks.



\commentOut{
metrics (y-axis): latency, throughput

1 Crash RB vs. Byzantine RB? 
    1, 2, 3, 4, 7, 8, 9 (0ms, 20ms)  -- done

2 Byz RB vs. Hash RB?

3 Hash RB vs. EC RB?
    1, 2, 3, 4, 8 (0ms, 20ms)  -- done

4 Vanilla vs. Chocolate? (Hash)
    1, 2 (0ms, 20ms)  -- done

5 Async. vs. Sync.? NCBA, Digest vs. Ours
    1-9 (0ms, 20ms)  -- done

6 Topology: single switch, linear, tree, fat-tree
    -- done
    
7 Impact of faulty behavior: equivocate or not?
    -- done
    
8 Impact of slow link
    2, 4, 8 (0ms, 20ms, 40ms)

9 BW limitation from source to others: optimal vs. others
    (x-axis: bw limitation)
    
10 Scalability: 4, 6, 8, 10    
    alternative: Scalability: 5, 7, 9
+++++++++=}
\end{abstract}

%% file: intro.tex
\section{Introduction}


We consider the reliable broadcast (RB) problem in an asynchronous message-passing system of $n$ nodes.
Intuitively, the
RB abstraction ensures that no two non-faulty nodes deliver different messages reliably broadcast by a source (or sender), either all non-faulty nodes deliver a message or none does, and that, if the source is non-faulty, all non-faulty nodes eventually deliver the broadcast message. Several fault-tolerant distributed applications~\cite{DuanRZ18,AttiyaCEKW19,CachinP02} require communication with provable guarantees on message deliveries (e.g., all-or-nothing property, eventual delivery property, etc.). 
Since Bracha's seminal work in 1987 \cite{Bracha:1987:ABA:36888.36891}, many Byzantine-tolerant RB (or simply Byzantine RB) protocols have been proposed, which improve metrics including bit complexity and round complexity. These results are summarized in Table \ref{table:result}. Here, we assume a message is of $L$ bits in size. For a detailed study of the related work please refer to Section \ref{s:related}.

To the best of our knowledge, none of the prior works have studied and evaluated the reliable broadcast (RB) protocols in a practical setting assuming reasonable local computation power and finite bandwidth. 
Toward this ends, we identify fundamental trade-offs, and use \textit{cryptographic hash functions} \cite{abs-1806-04437} and \textit{erasure coding} \cite{verapless_book}  to design more efficient algorithms. We also build a general evaluation tool on top of Mininet \cite{Mininet1, Mininet2}  to conduct a comprehensive evaluation under practical constraints. 
One goal of this paper is to provide guidance and reference for practitioners that work on fault-tolerant distributed systems. In particular, our results shed light on the following two questions:

\begin{itemize}
    \item Does there exist an RB protocol that achieves optimality in all four main metrics --  bit complexity, round complexity, computation complexity and resilience?
    
    \item What is the performance of RB protocols in a realistic network?
\end{itemize}



\subsubsection*{Motivation}


This work is motivated by the following observations when we tried to apply fault-tolerant RB protocols in practice:

\begin{itemize}
    \item Existing reliable broadcast mechanisms are \textit{not} efficient in terms of bandwidth usage and/or computation. (See Table \ref{table:result})
    
    \item Most RB protocols~\cite{Bracha:1987:ABA:36888.36891, BirmanJ87,ChangM84} assume unlimited bandwidth, and use flooding-based algorithms that send unnecessary redundant messages.
    
    \item RB and Byzantine consensus protocols (e.g.,
    \cite{Liang2012}, \cite{PatraR11}, \cite{abs-2002-11321}) that are proved to have optimal bit complexity usually have high round complexity and local computation.
    
    \item Theoretically speaking, if a protocol relies on a cryptographic hash function \cite{Preneel1CH} to ensure correctness, then it is \textit{not} always error-free, since it assumes that the adversary has limited computational power. However, we think this is acceptable in practical systems as many real-world systems use cryptographic hash functions, e.g., Bitcoin \cite{Nakamoto}.
    
    \item In many scenarios, the source may not reside in the system, and the bandwidth between the source and other nodes is usually more limited compared to bandwidth between two non-source nodes. For example, source could be a client or client proxy for a distributed storage system, and communicate with other nodes through Internet and non-source nodes communicate through highly optimized datacenter network. 
\end{itemize}

\subsubsection*{Main Contributions}
Motivated by our observations, we propose a family of algorithms that use hash function and erasure coding to reduce bit, round, and computation complexity.
\begin{itemize}
    \item \textit{EC-CRB}: Crash-tolerant erasure coding-based RB (Sec. \ref{app:EC-CRB})
    \item \textit{H-BRB}: Hash-based Byzantine-tolerant  RB (Sec. \ref{s:H-BRB})
    \item \textit{EC-BRB}: Byzantine-tolerant erasure coding-based RB (Sec. \ref{s:ec-BRB})
\end{itemize}
Table~\ref{table:result} provides a summary of our results and compares our results to prior work. Our EC-based protocols use $[n,k]$ MDS erasure codes. Please refer to Section \ref{s:MDS} for a preliminary on MDS codes.

Our Byzantine RB's bit complexity is listed as $O(nL + nfL)$, because in  cases when the source is non-faulty and the delay is small, the complexity is $O(nL)$. Only in unfortunate scenarios where the source equivocates, or some messages are lost, our protocols need to perform a recovery mechanism which incurs $O(fL)$ extra bits per node. We believe our protocols are appropriate in practice, as most systems assume small $f$. Our EC-based protocols have another advantage over other protocols; the bandwidth consumption between the source and other nodes is only $O(nL/k)$.

The rest of the paper is organized as follows: Section~\ref{section:model} introduces our models, notations and problem specification. In Sections \ref{s:H-BRB} and \ref{s:ec-BRB}, we present our main algorithmic results on Byzantine RB protocols. In Section~\ref{section:impossibility}, we present two impossibility results proving the optimality of our algorithms (in certain aspects). These impossibilities together also imply that there is \textit{no} RB protocol that achieves optimality in all four main metrics. Our benchmark framework, Reliability-Mininet-Benchmark (RMB), and evaluation results are detailed in Section~\ref{section:evaluation}. 


 

\begin{table*}[hptb!]
\begin{tabular}{|c|c|c|c|c|c|c|}
\hline
Algorithm                                                            & \begin{tabular}[c]{@{}c@{}}Bit\\ complexity\end{tabular} & \begin{tabular}[c]{@{}c@{}}System Size\\  (Resilience)\end{tabular} & \begin{tabular}[c]{@{}c@{}}Round \\ Complexity\end{tabular} & \begin{tabular}[c]{@{}c@{}} Error\\-free\end{tabular} & \begin{tabular}[c]{@{}c@{}}Uses \\ MDS  \\codes\end{tabular} & Bottleneck \\ \hline
\begin{tabular}[c]{@{}c@{}}
CRB \cite{Raynal18}\end{tabular}            &  $O(n^2L)$   & $\geq f+1$ & $1$ & Yes  & No & - \\ \hline
\begin{tabular}[c]{@{}c@{}}EC-CRB\end{tabular} 
  & $O(n^2L/k)$ & $\geq f+1$& 2  & Yes & Yes & MDS code \\ \hline

Bracha RB \cite{Bracha:1987:ABA:36888.36891} & $O(n^2L)$   & $\ge3f+1$   &  $3$    &  Yes  &      No   & Flooding \\ \hline
Raynal RB \cite{imbs2016trading} & $O(n^2L)$   & $\ge3f+1$  & $2$      &  Yes    &      No   & Flooding\\ \hline

Patra  RB \cite{PatraR11}& $O(nL)$ & $\geq 3f+1$ &  $9$ 
& Yes & Yes & \begin{tabular}[c]{@{}c@{}}Polynomial time \\local computation\\with large constants \end{tabular}\\ \hline
Nayak et al.\cite{abs-2002-11321}& $O(nL)$ & $\geq 3f+1$ &  $10$ 
& Yes & Yes & \begin{tabular}[c]{@{}c@{}}Polynomial time \\local computation\\with large constants \end{tabular}\\ \hline
H-BRB[3f+1] & $O(nL) + O(nfL)$ & $\geq 3f+1$& $3$  & No & No & Hash Function\\ \hline
H-BRB[5f+1] & $O(nL) + O(nfL)$ & $\geq 5f+1$& $2$  & No & No & Hash Function\\ \hline
EC-BRB[3f+1] & $O(nL) + O(nfL)$ & $\geq 3f+1$& $3$  & No & Yes & \begin{tabular}[c]{@{}c@{}}Hash Function \\+ MDS code \end{tabular}\\ \hline
EC-BRB[4f+1] & $O(nL) + O(nfL)$ & $\geq 4f+1$& $3$  & No & Yes & \begin{tabular}[c]{@{}c@{}}Hash Function \\+ MDS code \end{tabular}\\ \hline

\end{tabular}
\caption{\\Summary of our contributions and a comparison with previous work}
\label{table:result}
\end{table*}
    

%% file: model.tex
\section{Preliminaries}\label{section:model}
\subsection{Model and Notations}

We consider a static asynchronous message-passing system composed of a fully connected network of $n$ nodes, where up to $f$ nodes may be \textit{Byzantine} faulty. 

\paragraph{Network}
Nodes are sequential and fully connected by reliable and authenticated point-to-point channels in an asynchronous network. ``Asynchronous” means
that nodes do not have access to a global clock (or wall-clock time), and  each node proceeds at its own speed, which can vary arbitrarily with real time. 
Reliable channel ensures that (i) the network cannot drop a message if both sender and receiver are non-faulty, and (ii) a non-faulty node receives a message if and only if another node sent the message.  Authentication ensures that the sender of each message can be uniquely identified and a faulty node cannot send a message with a fake identity (as another node) \cite{Lynch96,Raynal18}.

In an asynchronous network, there is no known upper bound on the message delay. 
However, a message sent by a non-faulty node to another non-faulty node will eventually be delivered due to the reliability channel assumption.
When we say a node sends a message to all nodes, we assume that it also sends to itself. Note that this is achieved by performing multiple unicasts; hence, there is no guarantee on the delivery if the sender is faulty.

\paragraph{Fault Model} 
 A Byzantine node is a node
that behaves arbitrarily: it may crash, fail to send or receive messages, start in an arbitrary state, perform arbitrary state transitions, etc. 
A Byzantine node may have the power to \textit{equivocate}, i.e., send arbitrary messages to different sets of nodes. For example, when a Byzantine source node, $s$ sends a message, $m$ to all the nodes, it can equivocate and send a message $m_1$ to some nodes, a different message $m_2$ to some other nodes, and no message at all to the other nodes.  A node that exhibits a Byzantine behavior is also called faulty. Otherwise, it is  non-faulty. In our model, up to $f$ nodes can exhibit Byzantine behavior.

\commentOut{
\paragraph{Hash Function}
\lewis{Move hash part to Section 4.}
Some of our algorithms utilize an ideal cryptographic hash function, specifically with the following properties: 1) pre-image resistance (one-wayness) \sapta{add citations for these properties \lewis{We might not need it. It's pretty standard stuff. Maybe check how other papers did it. I don't think we need to list all three properties. Maybe just one-wayness?}} 2) collision resistance and 3) second pre-image resistance. 
Every node runs the exact same hash function. A good pick for such function would accelerate the hash computation and henceforth would accelerate the reliable broadcast node. 
Though SHA2 and SHA3 families of hash functions cannot be proven with these properties, they can be used in the practical sense. 

}
\paragraph{Notations}
Every message $m$ sent by a non-faulty source $s$ is associated with a sequence number or index 
$h$. Thus $m$ can be uniquely indexed through a tuple $(s, h)$ in the system due to the message authentication assumption discussed above. 
For example, in the distributed data store context, $h$ could be the key of the message or a sequence number associated with the message. 

In all of our algorithms, we use 
$\mset_i[s,h]$
to denote the set of messages that the node $i$ collects, in which are candidates that can be identified with $(s, h)$. 
When the context is clear, we omit the subscript $i$. 
We use 
$\ctr[*]$
to denote a local counter of certain type of messages that is initialized to 0.
We use $\mathbb{H}(*)$ to denote the cryptographic hash function. 

\begin{figure}
    \centering
    \includegraphics[width=0.6\textwidth]{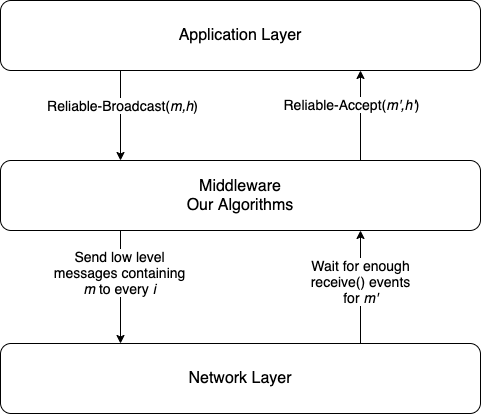}
    \caption{The reliable broadcast abstraction at a single node}
    \label{fig:layer}
\end{figure}

\subsection{Properties of Reliable Broadcast}\label{subsection:properties}

We adopt reliable broadcast properties from \cite{10.1007/11516798_17, Bracha:1987:ABA:36888.36891, DBLP:journals/corr/ImbsR15}. Each node consists of three layers: network, middleware (RB protocols), and application layers depicted in Figure \ref{fig:layer} (adapted from \cite{Raynal18}). The protocol at the source $s$ calls ``Reliable-Broadcast($m,h$)'' to broadcast a message $m$ with sequence number $h$  
reliably. Then, the middleware layer executes the RB protocol by exchanging messages with other nodes. For each non-faulty node, upon receiving enough messages of a certain type, the RB protocol would tell the application layer it can ``Reliable-Accept($m',h$)''. 

A reliable broadcast (RB) protocol is said to be \textit{correct} if it satisfies the following five properties.

\begin{property}[Non-faulty Broadcast Termination]
If a non-faulty source $s$ performs
Reliable-Broadcast$(m, h)$, with a message $m$ having index $h$ then all non-faulty nodes will eventually Reliable-Accept$(s, m, h)$.

\end{property}

\begin{property}[Validity]
If a non-faulty source $s$ does not perform Reliable-Broadcast$(m, h)$ then no non-faulty node will ever perform Reliable-Accept$(s, m, h)$.

\end{property}

\begin{property}[Agreement]
If a non-faulty node performs Reliable-Accept$(s, m, h)$ and another non-faulty node will eventually perform Reliable-Accept$(s, m', h)$ then $m = m'$.

\end{property}

\begin{property}[Integrity]
A non-faulty node reliably accepts at most one message of index $h$ from a source $s$.

\end{property}


\begin{property}[Eventual Termination]
If a non-faulty node performs Reliable-Accept$(s, m, h)$, then all non-faulty nodes eventually perform Reliable-Accept$(s, m, h)$.

\end{property}

Note that if the source is faulty, then it is possible that no non-faulty node would ever reliably accept its message. This is the main different between Byzantine RB problem, and Byzantine broadcast (or agreement) problem in the synchronous systems \cite{psl_BG_1982,Lynch96,Raynal18}. In the Byzantine broadcast problem, each non-faulty node has to output a value, whether the source is faulty or not.



%% file: crash1.tex
\section{Erasure Coding-based Crash-tolerant RB}
\label{app:EC-CRB}
In this section, we present a simple idea that augments the original crash-tolerant reliable broadcast (CRB) \cite{Lynch96} with erasure coding. The new protocol is named \textit{EC-CRB}.

\subsection{EC-CRB: Algorithm}

We present our algorithms here. EC-CRB will use $[n,k]$ MDS code. Source's logic is simple and its code is presented in Algorithm \ref{alg:EC-CRB-source}. To send a message $m$ with sequence number $h$, it encodes the message $m$ and then disseminates to each peer. The message is a tuple that contains the tag $\msg$, the source identifier $s$, corresponding coded element, and sequence number of the message $h$. 

To deal with asynchrony and failures, our algorithm is event-driven, similar to prior algorithms \cite{Bracha:1987:ABA:36888.36891,Raynal18}. 
The pseudo-code for peers and the source when receiving a message from the sender $j$ is presented in Algorithm \ref{alg:EC-CRB-peer}. 
First, upon receiving a coded element from the source, node $i$ forwards an $\echo$ message along with the coded element. Second, upon receiving an $\echo$ message, node $i$ decodes the message if it has received enough number of coded elements. 

The key design behind how crash-tolerant RB achieves the all-or-nothing property is that each peer needs to (pessimistically) help deliver the message to other peers. EC-CRB achieves this at Line \ref{line:crash-helper} in Algorithm \ref{alg:EC-CRB-peer}.
There are several designs that affect the complexity.

\begin{algorithm}
\caption{EC-CRB: source $s$ with message $m$ of index $h$}
\label{alg:EC-CRB-source}
\begin{algorithmic}[1]
\Function{Reliable-Broadcast}{$m, h$}
    \State $\{c_1, c_2, \dots, c_n\} \gets ENC(m)$ \Comment{Encoding message}
    \For{each $i$}
    \State {\sc Send}($\msg, s, c_i, h$) to node $i$
    \EndFor
\EndFunction
\end{algorithmic}
\end{algorithm}

\begin{algorithm}
\caption{EC-CRB: all node $i$ (including $s$) when receiving a message from sender $j$}
\label{alg:EC-CRB-peer}
\begin{algorithmic}[1]
\Function{Receiving}{$\msg, s, c, h$}
    \State {\sc Send}($\echo, s, c, h$) to all nodes 
    \State $\cset[s, h] \gets \cset[s,h] \cup \{c\}$
\EndFunction

\Function{Receiving}{$\echo, s, c, h$}
    \State $\cset[s, h] \gets \cset[s,h] \cup \{c\}$
    \If{$|\cset[s,h]| \geq k$ for the first time}
    \State
    $m \gets DEC(\cset[s,h])$ \Comment{Decoding}
    \State {\sc Reliable-Accept($s,m,h$)}
    \State
    {\sc Send} $(\acc, s, m, h)$ to all peers \label{line:crash-helper}
    \EndIf
\EndFunction
\end{algorithmic}
\end{algorithm}

\subsection{EC-CRB: Correctness and Complexity}

It is not difficult to see that EC-CRB is correct as long as $k \geq n-f$, since a node needs to have at least $k$ coded elements to correctly recover the original message, and in our model, a node can wait up to $n-f$ $\echo$ messages. As mentioned above, agreement property is achieved due to Line \ref{line:crash-helper} in Algorithm \ref{alg:EC-CRB-peer}.

The message complexity is $O(n^2)$. The round complexity is $2$. The bit complexity is $O(\frac{n^2 L}{k})$. For large enough $k$, the bit complexity becomes $O(nL)$.

%% file: byz.tex
\section{Hash-based Byzantine RB}
\label{s:H-BRB}

Crash-tolerant RB has been well-studied \cite{Lynch96,Raynal18}. For completeness, we present an erasure coding-based crash-tolerant RB (EC-CRB). This section focuses on Hash-based Byzantine-tolerant Reliable Broadcast (H-BRB) protocols.
In particularly, we present  H-BRB[3f+1] which uses a cryptographic hash function to reduce communication complexity. The name contains``3f+1'', because this protocol requires the system size $n \geq 3f+1$ for correctness. We also present the intuition of H-BRB[5f+1], which is correct if $n \geq 5f+1$. Compared to H-BRB[3f+1], it requires less number of rounds and messages. 

\subsubsection*{Byzantine Reliable Broadcast: Challenges}

We begin with the discussion on the difficulty of implementing a Byzantine RB protocol, and why most prior algorithms are \textit{not} practical due to prohibitively high bandwidth consumption. A Byzantine faulty node has a great deal of adversarial power. For example, it can equivocate and send out contradicting messages to different sets of nodes in the system. These nodes may collude to create a chain of misinformation and thus no information unless verified by at least $f+1$ nodes may be fully trusted. 

There are asymptotically tight algorithms in terms of either bit complexity or resilience or round complexity in the literature. Unfortunately many of them have high local computation~\cite{Patra11,PatraR11} 
and a large bandwidth consumption due to flooding of messages
~\cite{Bracha:1987:ABA:36888.36891,imbs2016trading}, which can be detrimental to practical networks with limited bandwidth. 

\subsubsection*{Cryptographic Hash Function}
All of our algorithms utilize an ideal cryptographic hash function. In particular, the correctness of our algorithms rely on the collision-resistant property of the hash function used. From a theoretical point of view, our algorithms are not error-free, as the adversary cannot have unlimited computation power. As discussed earlier, cryptographic hash functions are used widely in real-world applications. We believe it is reasonable to adopt this technique in designing more practical RB protocols.

Every node runs the same hash function. A good pick for such a function would accelerate the hash computation and henceforth accelerate the reliable broadcast process. By convention, the output of a hash function is of constant size.
Though SHA-2 and SHA-3 families of hash functions cannot be proven with these properties, they can be used in the practical sense.

\subsection{H-BRB[3f+1]}

\subsubsection*{H-BRB[3f+1]: Algorithm}

\begin{algorithm}
\caption{H-BRB[3f+1]: source $s$ with message $m$ of index $h$}
\label{alg:H-BRB_3f+1-source}
\begin{algorithmic}[1]
\Function{Reliable-Broadcast}{$m, h$}
    \State {\sc Send}($\msg, s, m, h$) to all nodes     \label{bo-rb-source-send-msg}
\EndFunction
\end{algorithmic}
\end{algorithm}

The pseudo-code of H-BRB[3f+1] is presented in Algorithms \ref{alg:H-BRB_3f+1-source}, \ref{alg:H-BRB_3f+1-peer}, and \ref{alg:H-BRB_3f+1-peer-check}. In Algorithm \ref{alg:H-BRB_3f+1-source}, the source node simply sends a $\msg$ message containing its identifier, message content $m$, and the sequence number $h$, to all the nodes. Following the convention, we assume that the source also sends the message to itself.

Algorithm \ref{alg:H-BRB_3f+1-peer} specified how all the nodes (including $s$) process incoming messages. Each node may receive five types of messages:

\begin{itemize}
    \item $\msg$ message: this must come directly from the source which contains the message content $m$. If the source identifier does not match the sender identifier, then the message is discarded.
    \item A \textit{helper($m$)} message is a constant sized message created from some arbitrary function $f(m)$. In our algorithms, the  function $\hash$ is used to create helper messages. The helper messages used in our algorithms are ECHO, ACC and REQ messages:
    \begin{itemize}
        \item  $\echo$ message: this message propagates information about a message already received by some node. In \cite{Bracha:1987:ABA:36888.36891,Raynal18}, $\echo$ messages contain the full content $m$. In our hash based algorithms, we only transmit $\hash(m)$. This is the main reason that we are able to reduce bit complexity.
    
        \item $\acc$ message: similar to  \cite{Bracha:1987:ABA:36888.36891}, this message is used to declare to other nodes when a some node is ready to accept a  message $m$. Again, instead of sending $m$ with the $\acc$ message, we send $\hash(m)$.
    
         \item $\req$ messages: In our hash based approach, a node might not know the original message $m$, even after it has observed enough $\acc(m)$  messages supporting it. 
         Therefore, such a node needs to use $\req(\hash(m))$ message to fetch the original message content from some non faulty node before accepting it.
    \end{itemize}
    
    
    
    \item  $\fwd$ messages:
     When a node is sent a $\req(\hash(m))$ message, it replies with a $\fwd(m)$ message, that contains the original message content of $m$. 
\end{itemize}



\begin{algorithm}[t]
\caption{H-BRB[3f+1]: all node $i$ (including $s$) when receiving a message from node $j$}
\label{alg:H-BRB_3f+1-peer}
\begin{algorithmic}[1]
\Function{Receiving}{$\msg, s, m, h$}
    \If{$j = s$\AND first $(\msg, s, *, h)$}              \label{bo-msg-first-msg}
        \State $\mset[s, h] \gets \mset[s, h] \cup \{m\}$       \label{bo-msg-recv-msg-tho-source}
        \State $\ctr[\echo, s, \hash(m), h] ++$
        \If{never sent $(\echo, s, *, h)$}                  \label{bo-msg-not-send-echo}
            \State {\sc Send}  ($\echo, s, \hash(m), h$) to all nodes
        \EndIf
    \EndIf
\EndFunction

\Function{Receiving}{$\echo, s, H, h$}
    \If{first ($\echo, s, *, h$) from $j$}             \label{ch-echo-fist-echo}
        \State $\ctr[\echo, s, H, h] ++$
        \State \textsc{Check}($s, H, h$)
    \EndIf
\EndFunction

\Function{Receiving}{$\acc, s, H, h$}
    \If{first ($\acc, s, *, h$) from $j$}              \label{ch-acc-first-acc}
        \State $\ctr[\acc, s, H, h] ++$
        \If{$\ctr[\acc, s, H, h] = f + 1$}      \label{ch-acc-send-req-tho-f+1-acc}
            \If{$\not \exists m' \in \mset[s, h]$ s.t. $\hash(m') = H$}
                \State {\sc Send}  ($\req, s, H, h$) to these $f + 1$ nodes 
            \EndIf
        \EndIf
        \State \textsc{Check}($j, H, h$)
    \EndIf
\EndFunction

\Function{Receiving}{$\req, s, H, h$}
    \If{first ($\req, s, h$) from $j$}                     \label{bo-req-first-req}
        \If{$\exists m' \in \mset[s, h]$ s.t. $\hash(m') = H$}
            \State {\sc Send}  ($\fwd, s, m', h$) to $j$
        \EndIf
    \EndIf
\EndFunction

\Function{Receiving}{$\fwd, s, m, h$}
    \If{have sent ($\req, s, \hash(m), h$) to $j$}         \label{bo-fwd-have-send-req}
        \If{first ($\fwd, s, m, h$) from $j$}              \label{bo-fwd-first-fwd}
            \State $\mset[s, h] \gets \mset[s, h] \cup \{ m \}$ \label{bo-fwd-recv-msg-tho-fwd}
            \State \textsc{Check}($s, \hash(m), h$)
        \EndIf
    \EndIf
\EndFunction

\end{algorithmic}
\end{algorithm}

\begin{algorithm}
\caption{H-BRB[3f+1]: helper function for all node $i$ (including $s$)}
\label{alg:H-BRB_3f+1-peer-check}
\begin{algorithmic}[1]

\Function{Check}{$s, H, h$}
    \If{$m \in \mset[s,  h]$  s.t. $\hash(m) = H$}\label{ch-check-msg-in-set}     
        \If{$\ctr[\echo, s, H, h] \geq f + 1$}
            \If{never sent ($\echo, s, *, h$)}
                \State {\sc Send}  ($\echo, s, H, h$) to all nodes
            \EndIf
        \EndIf
        \If{$\ctr[\echo, s, H, h] \geq n - f$}   \label{ch-check-send-acc-tho-n-f-echo}
            \If{never sent ($\acc, s, *, h$)}
                \State {\sc Send}  ($\acc, s, H, h$) to all nodes
            \EndIf
        \EndIf
        \If{$\ctr[\acc, s, H, h] \geq f + 1$}    \label{ch-check-send-acc-tho-f+1-acc}
            \If{never sent ($\acc, s, *, h$)}
                \State {\sc Send}  ($\acc, s, H, h$) to all nodes
            \EndIf
        \EndIf
        \If{$\ctr[\acc, s, H, h] \geq n - f$}    \label{ch-check-ra}
            \State \textsc{Reliable-Accept}($s, m, h$)
        \EndIf
    \EndIf
\EndFunction

\end{algorithmic}
\end{algorithm}



\subsubsection*{Correctness of H-BRB[3f+1]} 

\begin{theorem}
H-BRB[3f+1] satisfies Property 1-5 given that $n \geq 3f+1$.
\end{theorem}
 
We begin the proof with three important lemmas. The first two follow directly from the reliable and authenticated channel assumption and the thresholds we used.

\begin{lemma} If a non-faulty source $s$ performs Reliable-Broadcast($m,h$), then $\mset_i[s, h] \subseteq \{m\} $ at each non-faulty node $i$. 
\end{lemma}

\begin{lemma}
If a non-faulty node $s$ never performs Reliable-Broadcast ($m, h$), then $\mset_i[s, h] = \emptyset $ at each non-faulty node $i$.

\end{lemma}

\begin{lemma}
\label{lemma:H-BRB[3f+1]-acc}
	If two non-faulty nodes $i$ and $j$ send ($\acc, s, \hash(m), h$) and ($\acc, s, \hash(m'), h$) messages, respectively, then $m = m'$.
	
\end{lemma}

\begin{proof}
	Suppose, for the purpose of contradiction, $m \neq m'$. WLOG, let $i$ be the first node that sends ($\acc, s, \hash(m), h$), and let $j$ be the first node that sends ($\acc, s, H(m'), h$).
	Note that by construction, they are only able to send an $\acc$ message when either line \ref{ch-check-send-acc-tho-n-f-echo} or line \ref{ch-check-send-acc-tho-f+1-acc} of Algorithm \ref{alg:H-BRB_3f+1-peer-check} is satisfied.
	Since we assume $i$ is the first node that sends $\acc$ message in support of $m$, line \ref{ch-check-send-acc-tho-f+1-acc} of Algorithm \ref{alg:H-BRB_3f+1-peer-check} could not be satisfied. Therefore, $i$ must have received at least $ n - f~\echo$ messages supporting $\hash(m)$, out of which $n-f-f \geq f + 1~\echo$ messages are from non-faulty nodes. Thus, at most $f$ faulty nodes and  at most $f$ non-faulty nodes would send $\hash(m')$
	
	Now consider the case of node $j$. Since $j$ has received at least $n - f$ $\echo$ messages supporting $\hash(m')$, following the same rationale as above, at least $f + 1$ $\echo$ messages are from non-faulty nodes. This however leads to a contradiction, since the algorithm does not permit non-faulty nodes to 
	send $\echo$ messages supporting both $\hash(m)$ and $\hash(m')$.
\end{proof}

Property 1-4 follow directly from the three lemmas above. Below, we prove the most interesting one, Property 5 (Eventual Termination).

\begin{lemma}
    H-BRB[3f+1] satisfies Property 5 (Eventual Termination) if $n \geq 3f+1$.
\end{lemma}

\begin{proof}
	If a non-faulty node $i$ reliably accepts a message $m$, then predicates at line~\ref{ch-check-msg-in-set} and at line~\ref{ch-check-ra} of Algorithm~\ref{alg:H-BRB_3f+1-peer-check} are satisfied, which means $i$ has the message content $m$ (either directly from the source or a forwarded message) and $i$ has gathered at least $ n - f$ $\acc$ messages. Out of these $\acc$ messages, at most $f$ come from Byzantine nodes, and thus node $i$ has received at least $ n-f-f \geq f + 1$, $\acc$ messages from non-faulty nodes.
	
	By our assumption that all messages sent by non faulty nodes eventually reach non faulty nodes, at least $f + 1$ $\acc$ messages will eventually be delivered at all the other non-faulty nodes, and thus a non-faulty node can always receive some message broadcast by (possibly faulty) node $s$ since line~\ref{ch-acc-send-req-tho-f+1-acc} of Algorithm~\ref{alg:H-BRB_3f+1-peer} will be triggered.

	Line \ref{ch-check-msg-in-set} of Algorithm \ref{alg:H-BRB_3f+1-peer-check} is satisfied once a non-faulty node acquires the original message and 
	line \ref{ch-check-send-acc-tho-f+1-acc} of Algorithm \ref{alg:H-BRB_3f+1-peer-check} is satisfied when a node has $f + 1$, $\acc$ messages. Once this happens, by Lemma~\ref{lemma:H-BRB[3f+1]-acc}, each of the $n-f$ non-faulty nodes will send out the correct $\acc(\hash(m))$ message 
	to all the other nodes.  Thus, eventually $\geq n - f$ $\acc$ messages supporting $m$ will eventually be delivered to all the other non-faulty nodes, and the predicate on line \ref{ch-check-ra} of Algorithm \ref{alg:H-BRB_3f+1-peer-check} will be satisfied, and $m$ will be reliably delivered.  
\end{proof}

\subsection{H-BRB[5f+1]}


Inspired by a recent paper \cite{imbs2016trading} that sacrifices resilience for lower message and round complexity, we adapt H-BRB[3f+1] in a similar way. Particularly, we can get rid of the $\acc$ messages. By increasing number of servers, we are able to guarantee that after receiving $\geq n-f$, $ \echo$ messages, a node can reliably accept a message if $n\ge 5f+1$. Intuitively speaking, this guarantees that at least $ n-2f \geq 3f+1$ non-faulty nodes have received the same message (or more precisely the same $\hash(m)$), which is guaranteed to be a quorum that prevents other nodes to collect enough $\echo$ messages.



The pseudo-code of H-BRB[5f+1] is presented in Algorithm \ref{alg:H-BRB_3f+1-source}, \ref{alg:H-BRB_5f+1-peer}, and \ref{alg:H-BRB_5f+1-peer-check}. Note that the source code is the same as H-BRB[3f+1]. Correctness proof is similar to the ones in \cite{DBLP:journals/corr/ImbsR15}.

\begin{algorithm}[t]
\caption{H-BRB[5f+1]: all node $i$ (including $s$) when receiving a message from node $j$}
\label{alg:H-BRB_5f+1-peer}
\begin{algorithmic}[1]

\Function{Receiving}{$\msg, s, m, h$}
    \If{$j = s$\AND first $(\msg, s, *, h)$}              
        \State $\mset[s, h] \gets \mset[s, h] \cup \{m\}$      
        \State $\ctr[\echo, s, \hash(m), h] ++$
        \If{never sent $(\echo, s, *, h)$} 
            \State {\sc Send}  ($\echo, s, \hash(m), h$) to all nodes
        \EndIf
        
    \EndIf
\EndFunction

\Function{Receiving}{$\echo, s, H, h$}
    \If{first ($\echo, s, *, h$) from $j$}  
        \State $\ctr[\echo, s, H, h] ++$
        \If{$\ctr[\echo, s, H, h] = f+1$}
            \If{$\not\exists m \in \mset[s,  h]$  s.t. $\hash(m) = H$}
                \State {\sc Send}  ($\req, s, \hash(m), h$) to these $f+1$ nodes
            \EndIf
        \EndIf
        \State \textsc{Check}($s, H, h$)
    \EndIf
\EndFunction

\Function{Receiving}{$\req, s, H, h$}
    \If{first ($\req, s, h$) from $j$}                     
        \If{$\exists m' \in \mset[s, h]$ s.t. $\hash(m') = H$}
            \State {\sc Send}  ($\fwd, s, m', h$) to $j$
        \EndIf
    \EndIf
\EndFunction

\Function{Receiving}{$\fwd, s, m, h$}
    \If{have sent ($\req, s, \hash(m), h$) to $j$}         
        \If{first ($\fwd, s, m, h$) from $j$}              
            \State $\mset[s, h] \gets \mset[s, h] \cup \{ m \}$ 
            \State \textsc{Check}($s, \hash(m), h$)
        \EndIf
    \EndIf
\EndFunction

\end{algorithmic}
\end{algorithm}

\begin{algorithm}
\caption{H-BRB[5f+1]: helper function for all node $i$ (including $s$)}
\label{alg:H-BRB_5f+1-peer-check}
\begin{algorithmic}[1]

\Function{Check}{$s, H, h$}
    \If{$m \in \mset[s,  h]$  s.t. $\hash(m) = H$}        
        \If{$\ctr[\echo, s, H, h] \geq n - 2f$} 
            \If{never sent ($\echo, s, *, h$)}
                \State {\sc Send}  ($\echo, s, H, h$) to all nodes
            \EndIf
        \EndIf
        
        \If{$\ctr[\echo, s, H, h] \geq n - f$}
            \State \textsc{Reliable-Accept}($s, m, h$)
        \EndIf
    \EndIf
\EndFunction

\end{algorithmic}
\end{algorithm}

Our H-BRB[5f+1] protocol completes in $2$ rounds of communication among the  nodes which is one round of communication less compared to other protocols as proved in Theorem~\ref{theorem:2_round_impossibility} in Section~\ref{section:impossibility}. 

%% file: EC-byz.tex
\section{EC-based Byzantine RB}
\label{s:ec-BRB}

One drawback of the H-BRB is that the bit complexity or message size is still high. Especially, the source still needs to send $O(nL)$ bits.
One standard trick is to use erasure coding \cite{verapless_book} to reduce the message size. We present two ideas in this section. The key difference between our protocols and prior EC-based RB protocol \cite{Patra11} and our hash-based protocols is that these algorithms require the source to send its original message to all other nodes, whereas in our EC-based protocols, the source sends a small coded element.

\subsection{MDS Erasure Code: Preliminaries}
\label{s:MDS}

For completeness, we first discuss basic concepts and notations from coding theory.
We use a linear $[n, k]$  MDS (Maximum Distance Separable) erasure code~\cite{verapless_book} over a finite field $\mathbb{F}_q$ to encode the message $m$. An $[n, k]$ MDS erasure code has the property that any $k$ out of the $n$ coded elements, computed by encoding $m$, can be used to recover (decode) the original message $m$. 

For encoding, $m$ is divided into $k$ elements $m_1, m_2, \ldots, m_k$ with each element having  size $L/k$ 
(assuming size of $m$ is $L$). The encoder takes the $k$ elements as input and produces $n$ coded elements $c_1, c_2, \ldots, c_n$ as output, 
 i.e., 
 
 \[
 [c_1, \ldots, c_n] = ENC([m_1, \ldots, m_k]),
 \]
 
\noindent where $ENC$ denotes the encoder. For brevity, we simply use $ENC(m)$ to represent $[c_1, \ldots, c_n]$.
 
 The vector $[c_1, \ldots, c_n]$ is  referred to as the \textit{codeword} corresponding to the message $m$. Each 
 coded element $c_i$ also has  size $\frac{L}{k}$. 
 
 In our algorithm, the source disseminates one coded element to each node. We use $ENC_i$ to denote the projection of $ENC$ on to the $i^{\text{th}}$ output component, i.e., $c_i = ENC_i(v)$. Without loss of generality, we associate 
 the coded element $c_i$ with node $i$, $1 \leq i \leq n$.
 




\subsection{EC-BRB[3f+1]}

Our first idea is to adapt H-BRB[3f+1] so that each node $i$ not only forwards $\hash(m)$, but also a coded element $c_i$. 
This reduces bit complexity.
We use $[n, f+1]$ MDS erasure code, and do not use detection or correction capability.
In other words, the decoder function $DEC$ can correctly decode the original message if the input contains at least $f+1$ \textit{uncorrupted} coded elements. We do not need the correction/detection, because a node can use $\hash(m)$ to verify whether the decoded message is the intended one or not.


The pseudo-code of EC-BRB[3f+1] is presented in Algorithm \ref{alg:EC-BRB_3f+1-source}, \ref{alg:EC-BRB_3f+1-peer}, and \ref{alg:EC-BRB_3f+1-peer-check}. Note that Line \ref{line:computation} in Algorithm \ref{alg:EC-BRB_3f+1-peer} requires exponential computation.
The proof is similar to the ones for H-BRB[3f+1].

\begin{algorithm}[htbp!]
\caption{
EC-BRB[3f+1]: source $s$ with message $m$ of index $h$}
\label{alg:EC-BRB_3f+1-source}
\begin{algorithmic}[1]
\Function{Reliable-Broadcast}{$m, h$}
    
    \State $\{c_1, c_2, \dots, c_n\} = $ ENC$(m)$
    \State {\sc Send}($\msg, s, \hash(m), c_k, h$) to node $k$ 
\EndFunction
\end{algorithmic}
\end{algorithm}

\begin{algorithm}[htpb!]
\caption{EC-BRB[3f+1]: all node $i$ (including $s$) when receiving a message from node $j$}
\label{alg:EC-BRB_3f+1-peer}
\begin{algorithmic}[1]

\Function{Receiving}{$\msg, s, H, c, h$}
    \If{$j = s$\AND first $(\msg, s, *, *, h)$}
        \State $\cset[s, H, h] \gets  \cset[s, h, H] \cup \{c\}$ 
        \State $\ctr[\echo, s, H, h]++$
        \If{never sent $(\echo, s, *, h)$}
            \State {\sc Send}$(\echo, s, H, c, h)$ to all nodes
        \EndIf
    \EndIf
\EndFunction

\Function{Receiving}{$\echo, s, H, c, h$}
    \If{first ($\echo, s, *, *, h$) from $j$}  
        \State $\ctr[\echo, s, H, h] ++$
        \State $\cset[s, H, h] \gets  \cset[s, h, H] \cup \{c\}$ 
        \If{$\ctr[\echo, s, H, h] \geq f+1$}
            \If{$\not\exists m \in \mset[s,  h]$  s.t. $\hash(m) = H$}
                \For{each $C \subseteq \cset[s, H, j]$, $|C| = f+1$}\label{line:computation}
                    \State $m \gets DEC(C)$
                    \If{$\hash(m) = H$}
                        \State $\mset[$s, h$] \gets \mset[$s, h$] \cup \{m\}$
                    \EndIf
                \EndFor
            \EndIf
        \EndIf
        \State \textsc{Check}($s, H, h$)
    \EndIf
\EndFunction

\Function{Receiving}{$\acc, s, H, h$}
    \If{first ($\acc, s, *, h$) from $j$}            
        \State $\ctr[\acc, s, H, h] ++$
        \If{$\ctr[\acc, s, H, h] \geq f + 1$}      
            \If{$\not \exists m' \in \mset[s, h]$ s.t. $\hash(m') = H$}
                \State {\sc Send}  ($\req, s, H, h$) to nodes if have not sent  ($\req, s, H, h$) to them before
            \EndIf
        \EndIf
        \State \textsc{Check}($j, H, h$)
    \EndIf
\EndFunction

\Function{Receiving}{$\req, s, H, h$}
    \If{first ($\req, s, h$) from $j$}                     \label{bo-req-first-req}
        \If{$\exists m' \in \mset[s, h]$ s.t. $\hash(m') = H$}
            \State {\sc Send}  ($\fwd, s, m', h$) to $j$
        \EndIf
    \EndIf
\EndFunction

\Function{Receiving}{$\fwd, s, m, h$}
    \If{have sent ($\req, s, \hash(m), h$) to $j$}         \label{bo-fwd-have-send-req}
        \If{first ($\fwd, s, m, h$) from $j$}              \label{bo-fwd-first-fwd}
            \State $\mset[s, h] \gets \mset[s, h] \cup \{ m \}$ \label{bo-fwd-recv-msg-tho-fwd}
            \State \textsc{Check}($s, \hash(m), h$)
        \EndIf
    \EndIf
\EndFunction

\end{algorithmic}
\end{algorithm}

\begin{algorithm}[htpb!]
\caption{EC-BRB[3f+1]: helper function for all node $i$ (including $s$)}
\label{alg:EC-BRB_3f+1-peer-check}
\begin{algorithmic}[1]

\Function{Check}{$s, H, h$}
    \If{$m \in \mset[s,  h]$  s.t. $\hash(m) = H$}                        
        \If{$\ctr[\echo, s, H, h] \geq f + 1$}
            \If{never sent ($\echo, s, *, *, h$)}
                \State $\{c_1, \dots, c_n\} \gets ENC(m)$
                \State {\sc Send}  ($\echo, s, H, c_i, h$) to all nodes
            \EndIf
        \EndIf
        \If{$\ctr[\echo, s, H, h] \geq n - f$}   \label{ch-check-send-acc-tho-n-f-echo}
            \If{never sent ($\acc, s, *, h$)}
                \State {\sc Send}  ($\acc, s, H, h$) to all nodes
            \EndIf
        \EndIf
        \If{$\ctr[\acc, s, H, h] \geq f + 1$}    \label{ch-check-send-acc-tho-f+1-acc}
            \If{never sent ($\acc, s, *, h$)}
                \State {\sc Send}  ($\acc, s, H, h$) to all nodes
            \EndIf
        \EndIf
        \If{$\ctr[\acc, s, H, h] \geq n - f$}    \label{ch-check-ra}
            \State \textsc{Reliable-Accept}($s, m, h$)
        \EndIf
    \EndIf
\EndFunction

\end{algorithmic}
\end{algorithm}

The downside is that EC-BRB[3f+1] requires exponential computation. That is, it needs to find out the correct $f+1$ coded elements to decode, which requires $O(\binom{n}{f+1})$ computation. When $f$ is small, the computation is negligible. However, the scalability is limited.

\subsection{EC-BRB[4f+1]}

To fix the scalability issue, we rely on the correction capability of MDS code. Unfortunately, we have to sacrifice the resilience, and the algorithm only works when $n \geq 4f+1$. This trade-off turns out is necessary, as formally discussed in Section \ref{section:impossibility}.

\subsubsection*{Error-correcting MDS Codes}
In our setup,  we will use $[n, k]$ MDS code for 
$$k = n-3f$$
\noindent In other words, the distance between different codewords is $d = n-k+1 = 3f+1$. Our algorithm will rely on the following theorems from coding theory.

\begin{theorem}
\label{thm:coding}
The decoder function $DEC$ can correctly decode the original message if the input contains at least $n-f$ coded elements and among these used elements, up to $f$ may be erroneous.

    
\end{theorem}

\begin{theorem}
\label{thm:coding2}
Assume $n\geq 4f+1$. Consider codeword $C = $  \\$\{c_1, c_2, \dots, c_n\}$ and codeword $C' = \{c_1', c_2', \dots, c_n'\}$ such that (i) $C$ has at most $f$ erasures, (ii) $C'$ has at most $f$ erasures,\footnote{Erasures at $C$ and $C'$ may occur at different positions.} and (iii) at most $f$ of the remaining coded elements are different between the two codewords. If $DEC(C) = m$, then $DEC(C')$ either returns $m$ or detects an error.
\end{theorem}

Note that Theorem \ref{thm:coding2} does not work for $n \leq 4f$. This is because by construction, each pair of codewords has distance $3f+1$. Therefore if the source is faulty, it is possible to find a scenario that $DEC(C) = m$ and $DEC(C') = m'$ for $m' \neq m$ if $n \leq 4f$.




\subsubsection*{EC-BRB[4f+1]: Algorithm}

We present the peudo-code in Algorithms  \ref{alg:EC-BRB_4f+1-peer}, \ref{alg:EC-BRB_4f+1-peer-check}, and \ref{alg:EC-CRB-source}. The structure is similar to before. The key difference is that upon receiving the $\echo$ messages, each node uses decoder function to recover the original message $m$. If the source is non-faulty, then the error-correcting feature of MDS code trivially handles the corrupted coded element forwarded by other faulty nodes.  

The key challenge is to handle the colluding behaviors from Byzantine source and other nodes. For example, it is possible that some non-faulty node can correctly construct a message, but other non-faulty nodes are not able to. This is the reason that we need to have codeword distance at least $3f$.

Another aspect is that we use a plain RB, say Bracha's RB \cite{Bracha:1987:ABA:36888.36891} protocol, to reliably broadcast $\hash(m)$. This guarantees even if the source is faulty, non-faulty nodes cannot decode different values. Since $\hash(m)$ is a constant, it does not affect the overall bit complexity.

One interesting aspect is that the MDS code part takes care of some tedious check, so the logic in Algorithm \ref{alg:EC-BRB_4f+1-peer-check} is actually simpler. In particular, we do not need the rules for handling $\echo$ messages.

\begin{algorithm}[H]
\caption{
EC-BRB[4f+1]: source $s$ with message $m$ of index $h$}
\label{alg:EC-BRB_4f+1-source}
\begin{algorithmic}[1]
\Function{Reliable-Broadcast}{$m, h$}
    \State \textsc{Reliable-Broadcast}($
    \hashTag|\hash(m), h$)
    \State $\{c_1, c_2, \dots, c_n\} = $ ENC$(m)$
    \State {\sc Send}($\msg, s, c_k, h$) to node $k$ 
\EndFunction
\end{algorithmic}
\end{algorithm}

\begin{algorithm}
\caption{
EC-BRB[4f+1]: all node $i$ (including $s$) when receiving a message from node $j$}
\label{alg:EC-BRB_4f+1-peer}
\begin{algorithmic}[1]
\Function{Receiving}{$\msg, s, c, h$}
    \If{$j = s$\AND first $(\msg, s, *, h)$}
        \State $\cset[s, h] \gets  \cset[s, h] \cup \{c\}$ 
        \If{never sent $(\echo, s, *, h)$}
            \State {\sc Send}$(\echo, s, c, h)$ to all nodes
        \EndIf
    \EndIf
\EndFunction

\Function{Receiving}{$\echo, s, c, h$}
    \If{first ($\echo, s, *, h$) from $j$}
        \State $\cset[s, h] \gets \cset[s, h] \cup \{c\}$
        \If{$|\cset[s,h]| \geq n-f$}
        
            \State $m \gets$ DEC$(\cset[s,h])$
            \If{$m \neq $ ERROR}
                \State $\mset[s,h] \gets \mset[s,h] \cup \{m\}$
            
                \State \textbf{wait until} $\exists x \in \hset[s,h]$ s.t. $x = \hash(m)$
            
                \Comment{successful decoding}
                \If{never sent $(\acc, s, *, h)$ before}
                    \State {\sc Send}$(\acc, s, \hash(m), h)$ to all nodes
                \EndIf    
            \EndIf    
            
        \EndIf
    \EndIf
\EndFunction

\Function{Receiving}{$\acc, s, x, h$}
    \If{first ($\acc, s, *, h$) from $j'$} 
        \State $\ctr[\acc, s, x, h] ++$
        \If{$\ctr[\acc, s, x, h] = f+1$, and \\\hfill never sent $(\acc, s, *, h)$ before} 
            \State {\sc Send}($\acc, s, x, h$) to all nodes
            
        \EndIf
        \State Check($s,h$)
    \EndIf
\EndFunction

\Function{Receiving}{$\req, s, x, h$}
    \If{first ($\req, s, x, h$) from $j$} 
        \If{$\exists m \in \mset[s,h]$ s.t. $\hash(m) = x$}
            \State {\sc Send}  ($\fwd, s, m, \hash(m), h$) to $j$
        \EndIf
    \EndIf
\EndFunction

\Function{Receiving}{$\fwd, j, m, x, h$}
    \If{have sent ($\req, j, x, h$) to $j'$}  
        \If{first ($\fwd, j, m, x, h$) from $j'$, and
        $\hash(m) = x$}
            \State $\mset[j,h] \gets \mset[j,h] \cup \{m\}$
            \State Check($j,h$)
        \EndIf
    \EndIf
\EndFunction

\Function{Reliable-Accepting}{$\hashTag|H', h$} from source $j$
    \State $\hset[j,h] \gets \hset[j,h] \cup \{H'\}$
\EndFunction
\end{algorithmic}
\end{algorithm}

\begin{algorithm}
\caption{EC-BRB[4f+1]: helper function for all node $i$ (including $s$)}
\label{alg:EC-BRB_4f+1-peer-check}
\begin{algorithmic}[1]

\Function{Check}{$s, h$}
    \If{$\exists x \in \hset[s,h]$ s.t.\\\hfill $\ctr[\acc, s, x, h] \geq n - f$} 
        \If{$\exists m \in \mset[s,h]$ s.t. $\hash(m) = x$}
            \State {\sc Reliable-Accept}$(s,m,h)$
        \Else
            \State {\sc Send}  ($\req, s, \hash(m), h$) to these $n - f$ nodes
        \EndIf
        
    \EndIf
\EndFunction

\end{algorithmic}
\end{algorithm}

\subsubsection*{EC-BRB[4f+1]: Correctness}

We present proof sketch of the following theorem.

\begin{theorem}
H-BRB[4f+1] satisfies Property 1-5 given that $n \geq 4f+1$.
\end{theorem}

Property 1 (Non-faulty Broadcast Termination) is similar to before. The only two new aspects are: (i) non-faulty nodes will eventually be able to reliably accept $\hash(m)$; and (ii) message $m$ can be correctly constructed. (i) is due to the property of Bracha's protocol, and (ii) is due to the feature of MDS code.

Property 2 (Validity) and Property 4 (Integrity) can be proved similar, which are essentially due to the $n-f$ threshold at Line 9 and $f+1$ threshold at Line 19 of Algorithm \ref{alg:EC-CRB-peer}. 

Now, we show that EC-BRB[4f+1] satisfies Property 3 (Agreement). For a given $s, h$, suppose by way of contradiction, two non-faulty nodes $a$ and $b$ reliably accept two values $v_a$ and $v_b$, respectively. Suppose $v_a \neq v_b$. This means that $\ctr[\acc, s, \hash(v_a), h] \geq n-f$ at node $a$ and
$\ctr[\acc, s, \hash(v_b), h] \geq n-f$ at node $b$. By construction, each node only sends $\acc$ message once. This means that we have $2(n-f)-n \geq n-2f \geq f+1$ nodes that send contradicting $\acc$ messages, a contradiction.

Finally, we prove the following lemma.

\begin{lemma}
EC-BRB[4f+1] satisfies Property 5 (Eventual Termination).
\end{lemma}
\begin{proof}
Suppose a non-faulty node $u$ \textsc{reliable-accept}($s, m, h)$, which means it has $\ctr[\acc, s, x, h] \geq n-f$ and $\exists x \in \hset[s,h]$ s.t. $x = \hash(m)$. Therefore, all the other non-faulty nodes would have $\ctr[\acc, s, x, h] \geq n-2f \geq f+1$ and eventually reliably accept $x$.

Theorem \ref{thm:coding2} implies that it is impossible for a non-faulty node $i$ to decode a message $m' \neq m$ in round $h$ for source $s$. Therefore, we have the following claim.

\begin{claim}
\label{claim:EC-BRB[4f+1]-acc}
If two non-faulty nodes $i$ and $j$ sends ($\acc, s, \hash(m), h$) and ($\acc, s, \hash(m'), h$) messages, respectively, then $m = m'$.
\end{claim}


Line 19 of Algorithm \ref{alg:EC-BRB_4f+1-peer} and Claim \ref{claim:EC-BRB[4f+1]-acc} ensure that all non-faulty nodes will eventually send an $\acc$ message advocating $x = \hash(m)$.
Therefore, eventually $\ctr[\acc, s, x, h] \geq n-f$ at every non-faulty node. At that point, each non-faulty node will either request a value from a set of nodes, particularly from $u$, or have collected enough correct coded elements to decode the value $m$. 
\end{proof}

\subsection{Discussion}

It is natural to ask whether it is possible to have an algorithm that does \textit{not} require exponential local computation or flooding of the original message, and requires only $3f+1$ nodes. We briefly discuss why such an algorithm is difficult to design, if not impossible. EC-BRB[4f+1] requires at least $4f+1$ nodes, as we require codeword distance to be $3f+1$. It is because that when $n=4f$, the adversary (including equivocating source) is able to force a group of non-faulty nodes to decode $A$ and the other to decode $B \neq A$, which defeats the purpose of using error-correcting code. Essentially, this creates the split-brain problem. Similarly, if we use codeword distance less than or equal to $3f$. Then it is again possible to force two group of non-faulty nodes to decode different values. In other words, techniques other than error-correcting codes need to be used to address this issue.

%% file: Impossibility.tex
\section{Impossibility Results}\label{section:impossibility} 

\subsection{Preliminaries}




To facilitate the discussion of our impossibility proof, we introduce formalized definition of our systems and useful notions such as time, execution, and phase. Note that these notions are well studied in the distributed computing literature~\cite{Lynch96,AttiyaW2004}. 
We include them here for completeness.

The system is made up of a finite, non-empty static set $\Pi$ of $n$  asynchronous \textit{nodes}, 
of which at most $f$ maybe \textit{Byzantine} faulty. 
An asynchronous node can be modeled as a state machine with a set of states $\mathcal{S}$. Every node $p$ starts at some initial state $s_p^i \in \mathcal{S}$. We consider event-driven protocols. That is, state transitions are triggered by the occurrences of events. Possible triggering events for node $p$ are: Reliable-Broadcast request ({\sc Reliable-Broadcast}$_p$) and receiving a message ({\sc Receiving}$_p$). 
A \textit{step} of a node $p$ is a 5-tuple $(s, T, m, R, s')$ where $s$ is the old state, $T$ is the triggering event, $m$ is the message to be sent, $R$ is a response ({\sc Send}$_p$ and {\sc Reliable-Accept}$_p$) or $\perp$, and $s'$ is the new state. The values of $m$, $R$ and $s'$ are determined by a transition function\footnote{The transition function is determined by the protocol or algorithm} applied to $s$ and $T$. The response to {\sc Reliable-Broadcast} is {\sc Send}$_p$, and  the response to {\sc Receiving}$_p$ is either {\sc Send}$_p$ or {\sc Reliable-Accept}$_p$. 
%
If the values of $m$, $R$, and $s'$ are determined by the node’s transition function applied to $s$ and $T$, then the step is said to be \textit{valid}. 
In an \textit{invalid }step (taken by Byzantine node), the values of $m$, $R$, and $s'$ can be arbitrary. 

A \textit{view} of a node is a sequence of steps such that:
(i) the old state of the first step is an initial state;
   (ii) the new state of each step equals the old state of the next step.
%
A point in \textit{time} is represented by a non-negative real number. 
A \textit{timed view} is a view whose steps occur at non-decreasing times. If a view is infinite, the times at which its steps occur must increase without bound. %
If a message $m$ sent at time $t$ is received by a node at time
$t'$, then the delay of this message is $t'-t$. This encompasses transmission delay as well as time for for handling the message at both the sender and receiver. 

An \textit{execution} $e$ is a possibly infinite set of timed views, one for each node that is  present in the system, that satisfies the following assumptions:
\begin{enumerate}
    \item Every message \textsc{Send} has at most one matching \textsc{Receive} at the receiving node and every message \textsc{Receive} has exactly one matching message \textsc{Send}.
    \item If a  non-faulty node $p$  \textsc{Send}'s message $m$  to a non-faulty node $q$ at time $t$, then $q$ \textsc{Receive}'s message $m$ at finite time $t' \ge t$ 
\end{enumerate}
An execution is said to be \textit{legal} if at least $n-f$ timed views in the execution are valid. We consider an algorithm to be correct if every execution of the algorithm satisfies properties 1 to 5 in Section~\ref{subsection:properties}.

Now, we introduce some notions with respect to the broadcast problem on any execution prefix $e$:
\begin{itemize}
    \item The  maximum message delay  in any execution prefix $e$ is called the $phase$. 
    \item If a node $p$ hears about a message directly from the sender, it is a {\em direct witness}. The sender of a message is also a direct witness.
    
    \item A $witness(m)$ message is a message sent by a node $p$ announcing that $p$ has either received $m$ directly from the sender or heard of $f+1$ nodes that have sent $witness(m)$ messages supporting $m$. The sender can also send out $witness(m)$ messages. Nodes send out witness messages to themselves as well. We assume that all witness messages that other than the constant size helper messages (in Section~\ref{s:H-BRB}) contain the original message piggybacked onto it. 
    
    \item Node $p$ is an {\em indirect witness} of message $m$, if $p$ hears of at least $f+1$ witness($m$) messages. 
    \item A \textit{helper($m$)} message is a constant sized messages created from some arbitrary function $f(m)$. In the sections above, the hash function is used to create helper messages. 
\end{itemize}
	
In reliable broadcast algorithms presented in this paper and other prior implementations~\cite{Bracha:1987:ABA:36888.36891, Raynal18, DBLP:journals/corr/ImbsR15}, $witness(m)$ messages can be in the form of Echo, Ready, Accept, etc. For our impossibility proofs, we consider the family of algorithms in which when nodes become a witness to some message  $(m,h)$, they have to 
send out a witness message supporting $(*,h)$. 
Since we consider an asynchronous system, nodes do not have the capability to measure time and thus the algorithm is event-driven. That is, the algorithm logic has to be based on waiting for a certain number of messages before making a decision.

For all event-driven witness-based algorithms, the following observation applies. The reason is that Byzantine nodes can choose not to send out witness messages, and if an algorithm needs to wait for more than $n-f$ witness messages to proceed, then a non-faulty node could never make progress.
\begin{observation}~\label{Obsv:Sufficient}
	A node cannot wait for $>n-f$ witness messages for message $m$ to reliably deliver $m$.
\end{observation}

\subsection{Impossibility: Byzantine RB in 2 phases}


For Lemmas~\ref{Obsv:Necessary} and~\ref{lem:Uniqueness_1} and Theorem~\ref{theorem:2_round_impossibility}, consider a system $\Pi$ of size $|\Pi|= n = 5f$, composed of node sets $ S_1, S_2, S_3,S_4$ and $B$. Each of these sets is of size $f$, i.e., $|S_1| = |S_2|= \cdots = |B| = f$. Assume set $B$ to be the set of Byzantine faulty nodes. For all the proofs below we consider an execution prefix and the maximum message delay or phase in each prefix is considered to be $1$ time unit. Note that in these proofs, all messages (unless mentioned otherwise) have a fixed message delay of $1$ unit of time. In this case, by definition, phase $r$ refers to the time interval $[r,r+1)$. 

\begin{restatable}{lemma}{obsvnecessary}~\label{Obsv:Necessary}
	A node must wait for $\ge \floor{\frac{n+f}{2}}+1$ witness messages for message $m$, to reliably deliver $m$.
	\begin{proof}
    Assume by contradiction Algorithm $\mathbb{A}$ reliably delivers message $m$ at node $p$ after  receiving $\floor{\frac{n+f}{2}}$ $witness(m)$ messages at $p$.
    
    We now consider a scenario where Property 3 (Agreement) is violated. 
    Node $b\in B$ equivocates and sends out message $m_1$ to nodes in set $S_1\cup S_2$ and message $m_2$ to nodes in $S_3\cup S_4$, 
    associated with index $h$. Each non-faulty node sends out a witness message as per the message received from $b$ with index $h$. Since the sender $b$ is also a direct witness, it is allowed to send out witness messages. All the Byzantine nodes in $B$ including node $b$ equivocate and send out $witness(m_1)$ to the nodes in $S_1\cup S_2$ and $witness(m_2)$ to the nodes in $S_3\cup S_4$. As a result, nodes in $S_1\cup S_2$ get $3f$, $witness(m_1)$ messages and  $2f,$ $witness(m_2)$ messages and nodes in $S_3\cup S_4$ get $3f$ $witness(m_2)$ messages and $2f$ $witness(m_1)$ messages. 
    
    Note that in this system, 
    $\floor{\frac{n+f}{2}} = \floor{\frac{5f+f}{2}} = 3f$. By the statement of the lemma, nodes wait for only $3f$ $witness(m)$ messages to deliver any message $m$. Thus nodes in $S_1\cup S_2$ deliver message $m_1$ and nodes in $S_3\cup S_4$ deliver message $m_2$  with respect to index $h$, violating Property 3 (Agreement).


\end{proof}

\end{restatable}

\begin{restatable}{lemma}{uniquenessOne}\label{lem:Uniqueness_1}
	If non-faulty nodes send contradicting witness messages corresponding to a message with index $h$, Property 3 (Agreement) is violated. 
	\begin{proof}
	The proof is by contradiction. Assume Algorithm $\mathbb{A}$ ensures reliable broadcast.  
	 Consider a prefix of an execution we term Exec-1 as follows.\\ 
	\textbf{{Exec-1:}}
	\begin{itemize}
		\item \textbf{Phase $r$:} Assume node $b \in B$ sends message $m_1$ to nodes in  $S_1$ and message $m_2(\ne m_1)$ to node $S_4$ in phase $r$. 
		Nodes in $S_1$ and $S_4$ hear about  messages $m_1$ and $m
		_2$ respectively before the beginning of the next phase. 
		\item \textbf{Phase $r+1$:} In this phase, the following messages are sent out.\\
		\textit{Byzantine nodes in $B$:}
		\begin{enumerate}
		    \item All the nodes in  $B$ send out $witness(m_1)$ to nodes in $S_1$ and $witness(m_2)$ to nodes in  $S_2$. 
		    \item Node $b\in B$ sends out $witness(m_1)$ to all nodes in $S_2\cup S_3$. These witness messages are fast and take time $t' \ll 1$  and reach the nodes in $S_2\cup S_3$ at time $(r+1)+t'$.
		    \item Nodes $B-\{b\}$ send out $witness(m_2)$ messages to all nodes in $S_2\cup S_3$. These messages are also fast but slower than the $witness(m_1)$ messages sent by $b$  and take $2\cdot t'$ time units where, and reach the nodes in $S_2\cup S_3$ at time $(r+1)+2t'$. Note that $2\cdot t' \ll 1$. 
		\end{enumerate}  
 		
 		\vspace{3pt}
 		
		\textit{Nodes in $S_1$:}\\
		Nodes in  $S_1$ send out $witness(m_1)$ message to all nodes. These messages are very fast and take time $t' \ll 1$  and reach the nodes in $S_2\cup S_3$ at time $(r+1)+t'$. 
		
		\vspace{3pt}
		
		\textit{Nodes in $S_4$:}\\
		Nodes in  $S_4$ send out  $witness(m_2)$ messages to  all nodes.  These messages are fast but slower than the $witness(m_1)$ messages sent by nodes in $S_1$  and take $2\cdot t'$ time units; where,  $2\cdot t'\ll 1$. 
		\\\textit{Nodes in $S_2\cup S_3$:}\\Consider the time interval $[r+1+t', r+1+2\cdot t']$. 
		 At time $r+1+t'$, these nodes receive $f+1$ messages in the form of $witness(m_1)$, and become indirect witnesses for message $m_1$. As a result they send out $witness(m_1)$ messages to everyone. The messages to the nodes in $S_1$ are very fast and have a delay of $t' \ll 1$ time units but reach everyone else slowly, with a delay of $1$ time unit. 
		 
		 At time $r+1+2t'$, nodes in $S_2\cup S_3$ receive $2f-1$ messages of the form $witness(m_2)$ from nodes in $S_4\cup B -\{b\}$.  At this point, these nodes also become indirect witnesses for message $m_2$ and send out witness messages of the form $witness(m_2)$ to everyone. By the assumption on Algorithm $\mathbb{A}$, this is legal and the prefix of the execution is still valid. These witness messages are fast (delay = $t'$) and reach all nodes by time $r+1+3t'$. 

		
	\end{itemize}
	\begin{itemize}
		\item \textbf{Phase $r+2$:} \\
		By time $r+2$, all nodes in $S_1$ have received $n-f$ $witness(m_1)$ messages from $B\cup S_1\cup S_2 \cup S_3$ and all nodes in $S_4$ have received $n-f$ $witness(m_2)$ messages from $S_2\cup S_3 \cup S_4 \cup B$. By Observation~\ref{Obsv:Sufficient}, nodes in $S_1$ and $S_4$ cannot wait for more witness messages and thus deliver messages $m_1$ and $m_2$, respectively, violating Property 3 (Agreement). 
		
	\end{itemize}

	Thus a non faulty node may not send out contradicting witness messages for any given message index. 

\end{proof}
	\end{restatable}

\begin{theorem}
	\label{theorem:2_round_impossibility}
	It is impossible to guarantee reliable broadcast in 2 phases of communication after a non-faulty node hears about a broadcast message. 
\end{theorem}
\begin{proof}
We do a proof by contradiction. 
Assume Algorithm $\mathbb{A}$ ensures reliable broadcast in 2 phases of communication after a non-faulty node hears about this message. Recall that all messages have a delay of $1$ time unit unless otherwise mentioned and phase $r$ refers to the time interval $[r,r+1)$. Consider a prefix of an execution we term Exec-2 as follows. 
\\\textbf{{Exec-2:}}
\begin{itemize}
	\item \textbf{Phase $r$:} Let $b\in B$ equivocate and send a message $m_1$ to nodes in $S_1$ and message $m_2(\ne m_1)$ to nodes in $S_2$ 
	in phase $r$. Nodes in $S_1$ and $S_2$ are direct witnesses before the beginning of the next phase. 
	\item \textbf{Phase $r+1$:}\\ Nodes in $B$ collude and   send contradicting witness messages as follows: $witness(m_1)$ to nodes in $S_1\cup S_3\cup S_4$ and $witness(m_2)$ to nodes in $S_2$.  Nodes in $S_3\cup S_4$ do nothing.
	
	Lemma~\ref{lem:Uniqueness_1} shows that, nodes cannot be witnesses to more than one message per sender as this will violate Property 3 (Agreement). As a result nodes in $S_1$ only send $witness(m_1)$ message to all nodes, and nodes in  $S_2$ only send $witness(m_2)$ message to  all nodes.

	\item \textbf{Phase $r+2$:}\\
	Nodes in $S_3\cup S_4$ send out $witness(m_1)$ messages to all nodes as they are indirect witnesses to message $m_1$. 
	
	By the correctness of $\mathbb{A}$ and our assumption at the beginning of the proof, all nodes reliably deliver a message, either  $m_1$ or $m_2$ by the end of this phase. 
	WLOG, let the message that each node delivers is $m_1$. Thus all nodes should have received at least $(n+f)/2 +1 = 3f+1$ $witness(m_1)$ messages (from Lemma~\ref{Obsv:Necessary}) and at most 
	$n-f = 4f$ $witness(m_1)$ messages (from Observation~\ref{Obsv:Sufficient}) before reliably delivering it. 
	
	However, by the end of phase $r+2$, 
	any node $s_2\in S_2$ ends up not getting $\ge 3f+1$ messages supporting either $m_1$ or $m_2$. 
	Node $s_2$ receives  at most $2f$ $witness(m_2)$ messages (from $B$ and $S_2$) and at most $3f$ $witness(m_1)$ messages (from $S_1\cup S_3 \cup S_4$) by the end of phase $r+2$. 
	
	Node  $s_2$ thus has to wait on $2f$ more $witness(m_2)$ messages or $f$ more $witness(m_1)$ messages to satisfy the sufficiency condition in Observation~\ref{Obsv:Sufficient}, which do not arrive after at the end of phase $r+2$. 
\end{itemize}

This violates Property 5 (Eventual Termination).
\end{proof}

\subsection{Impossibility of Byzantine RB using constant-size helper messages in $4$ phases} 
Here, we consider algorithms where the non-source sends out constant-size \textit{helper} messages to ensure reliable broadcast except  
when a receiving node specifically requests for the original message. A node $j$ does not send out requests for the original message until it is absolutely sure that at least one non-faulty node has reliably delivered a message $m$ associated with some tag $h$ different from the message ($m',h$), $j$ has received. Note that $j$ might not even have  received $m'$, i.e., $m' = \perp$. 

\begin{restatable}{theorem}{impossibility}\label{thm:impossibility2}
	For $n \leq 5f+1$, it is impossible to guarantee reliable broadcast in $4$ phases of communication if using constant sized helper messages,  once a non-faulty node hears about a broadcast message . 
	\begin{proof}
Assume a system with size $n =6$ containing nodes $\{1,2,3,4,5,b\}$ where node $b$ is byzantine faulty. Consider the following execution:
\begin{itemize}
\item \textbf{Phase $r$:}  Node $b$ performs a Reliable-Broadcast of message $(m,h)$ to all nodes. However, it equivocates and sends $(m',h)$ to node $1$ and $(m,h)$ to the others. 
\item \textbf{Phase $r+1$:}  Since we are considering algorithms that send out messages of constant size for non-senders, unless specifically requested, 
node $1$ sends out \textit{helper($m'$)} messages for $m'$ and nodes $\{2,3,4,5\}$ send out \textit{helper($m$)} messages for $m$. Node $b$ equivocates again and sends out \textit{helper($m'$)}  to node $1$ and \textit{helper($m$)} to all others.
\item \textbf{Phase 2:} Nodes $2,4,4,5$ have now heard \textit{helper($m$)} messages supporting $m$ from $5 = n-f$ nodes. We know from Observation~\ref{Obsv:Sufficient}, that a node cannot wait for more messages to reliably deliver a message. Thus these four nodes reliably deliver $m$ and send out helper messages of acceptance. Node $1$ has heard of only one original message $m'$ and \textit{helper($m$)} messages about message $m\ne m'$. Thus it is unable to deliver any message associated with the tag $h$ at this moment. Note that this is the first round when node $1$ hears about the existence of a message different from $m'$ associated with tag number $h$. 
\item \textbf{Phase $r+3$:}  Node $1$ hears of accept messages from $n-2f$ nodes (from $2,3,4,5$) about a message $m \ne m'$. Thus, since more than $f$ nodes have accepted a different message that node $1$ has not received yet, node $1$ sends out a request for the original message to any $f+1$ nodes that sent out the accept messages. WLOG, let node $1$ send request messages to nodes $2$ and $3$.
\item \textbf{Phase $r+4$:}  Nodes $2$ and $3$ send out the original messages as a response to $1$'s request. Node $1$ is still unable to reliably deliver any message associated with the tag $h$. 
\end{itemize}
This violates Property 5 (Eventual Termination) in Section~\ref{subsection:properties}. 

In the next phase $r+5$, node $1$ sends out a decision to accept the original message $m$ after receiving it from nodes $2$ and $3$. Thus in phase $r+5$ node $1$ finally reliably delivers $m$. 
\end{proof}

\end{restatable}

%% file: eval.tex
\newcommand{\framework}{RMB}
\newcommand{\frameworkSpace}{\framework\hspace{2pt}}
\section{Evaluation}\label{section:evaluation}

We evaluate the performances of RB protocols through simulations over realistic environments. 
We design and implement a configurable and extensible benchmarking platform for reliable protocols over asynchronous message-passing networks, \textit{Reliability Mininet Benchmark (\framework)}. In particular, RMB is appropriate for evaluating protocols over networks within a datacenter or a cluster.
In this section, we describe the framework, including its features and parameters first, then report performance numbers of the proposed protocols and some prior algorithms using \framework.


\subsection{The Architecture of \framework}


RMB is built on top of Mininet \cite{Mininet1,Mininet2}, and tailored toward the generic abstraction for distributed algorithms depicted in Figure \ref{fig:layer}. The architecture of RMB is presented in Figure \ref{fig:framework2}. 
The application in RMB is a workload generator that generate reliable broadcast requests and collect and report statistics. The protocol layer contains the RB protocols that we implement. RMB is extensible in the sense that as long as the protocol implementation follows the pre-defined interface, the implementation can be evaluated using RMB. Finally, the network layer is simulated by Mininet and the network manager that we implement. 
RMB users can easily use the script we provide to configure the network conditions, e.g., delay, jitter, bandwidth, etc.
By design, RMB components (application, protocol, and network) are run in separate processes, and the entire RMB is simulated on a single machine using Mininet. The layers are implemented using Go\footnote{https://golang.org/} and we provide Python scripts to launch RMB.

Network parameters can be easily configured within a single YAML file. Benchmark managers and protocols are invoked by a Python script after the network is initialized. Advanced users can also program a topology and try various link parameters. A key benefit of \frameworkSpace is to free user from tedious work of configuring environments (regarding networks and computation power) and faulty behaviors. 
\subsection{RMB Benchmark Workflow}
\label{app:RMB-workflow}

\frameworkSpace first reads network parameters in a configuration YAML file to create a Mininet network with a preset topology (one of SingleSwitchTopo, LinearTopo, TreeTopo, and FatTreeTopo, or a user-defined one), with the desired number of hosts (with CPU limit or not) and types of links (e.g., with different artificial delay, packet jitters, bandwidth constraints).

We have programmed reliable broadcast protocols mentioned in previous sections with respect to the benchmark manager's interface. After specifying binary files of these two programs and filling in parameters required by them in the YAML file, a user of \frameworkSpace can starts all three layers of \frameworkSpace with a single command (that fires up the Python script). After a round of simulation is complete, each benchmark manager would generate statistics that can be later analyzed and graphed through analytical scripts of \framework.


\begin{figure}[h]
\centering
\includegraphics[trim={4cm 2cm 0.5cm 1cm},clip, width=0.47\textwidth]
{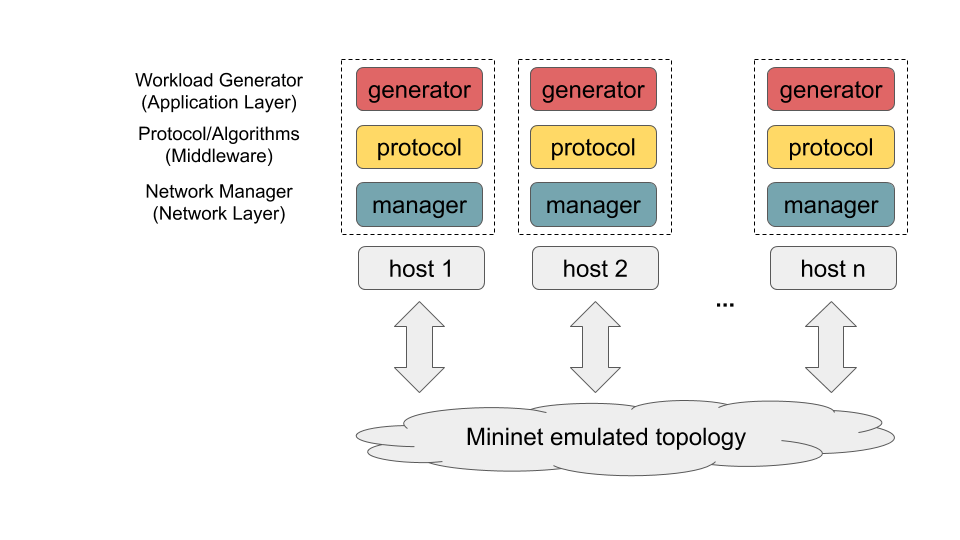}
\caption{The architecture of \framework}
\label{fig:framework2}
\end{figure}

\subsubsection*{Mininet}
The bottom layer (gray boxes in figure \ref{fig:framework2}) is a virtualized network, created by Mininet \cite{Mininet1, Mininet2}. Mininet is a battle-tested software that is widely used in prototyping Software-Defined Networks (SDNs). Our python starts-up script calls the Mininet library to start a virtual network consist of hosts, links, switches and a controller before the start of simulation. 
A virtual host emulates a node in a networked computing cluster, but in essence, it is a shell in a network namespace. Other three layers of \frameworkSpace are Linux applications that run on each host inside Mininet. Hosts do not communicate with each other directly, instead, they connect to virtual switches through Mininet links (Linux virtual Ethernet device pairs). The switches are also connected through links, if there are more than one. These controller and switches together make a good emulation of hardware in a real-world network.

We choose to use virtualized network owing to several benefits, including scalability (e.g., no limit on number of ports for a virtual switch), almost-zero setup time (i.e., a command can start the entire network), fine-grained control over network artifacts (e.g., delay and jitter), and no concern with network bandwidth fluctuations.

\subsubsection*{Manager Layer}
We have one (network) manager for each host is to manage data communication between protocol layer and other hosts in the network. There are four go routines for separate responsibilities: (i) receiving message from the protocol layer, (ii) receiving message from other hosts, (iii) sending to protocol layer, and (iv) sending message to other hosts. Another responsibility is to control the faulty behavior if the current node is configured to be Byzantine node, e.g., randomly corrupt messages. 

\subsubsection*{Protocol Layer (RB Algorithms)}
The middle layer implements the RB protocols  we want to evaluate. 
Each instance is paired up with a manager we discussed above, and thus does not need to know explicitly the existence of other manager/protocol instances. Such a design choice allows researchers to implement new protocols and benchmark them at ease.
In our RB protocols, there are two go routines in this layer. One is responsible for sending message to the manager layer, and the other one is responsible for reading messages from the manager layer and then perform corresponding action. That is, we implemented an event-driven algorithm as in our pseudo-code. Note that we make minimal assumption in this layer; hence, potentially, future RMB users can implemented in their favorite programming language and the algorithms do not have to be event-driven. For the hash function, we used Golang default package hmac512\footnote{https://golang.org/pkg/crypto/hmac/}, and open source Golang erasure coding package written by klauspost \footnote{https://github.com/klauspost/reedsolomon}.

\subsubsection*{Application Layer (Workload Generator)}
The top layer implements the workload generator in RMB.
There are two roles: (i) issue reliable-broadcast commands following a specified workload (e.g., size, frequency), and (ii) collect and calculate statistics (latency and throughput).


\subsection{Performance Evaluation}


\subsubsection*{Simulation Setup}

We perform the performance evaluation using \frameworkSpace on a single virtual machine (VM), equipped with 24 vCPU and 48 GB memory on Google Cloud Platform (GCP). Ubuntu 18.04.4 LTS (Bionic Beaver) runs on the VM as the OS. 
When no network artifact is specified, the default RTT in a Mininet emulated topology, from a host to another, is between 0.06 ms to 0.08 ms.

In the evaluation below, we do not present the result for Patra's algorithm \cite{Patra11}, which has optimal complexity $O(nL)$. The reason is its high computation complexity. It requires roughly $400$ ms to complete for a network of $32$ nodes, which makes it much slower than the other protocols we test.


\subsubsection*{Evaluation: Topology and Bandwidth Constraints}
RMB allows us to easily evaluate our algorithms in different topology and bandwidth constraints.
We test in three different network topologies with
$n=5$ and $f=0$: (i) \textit{Linear topology}: 5 switches with one host per switch; (ii) \textit{Tree topology}: tree depth = 3, and fan-out = 2; and (iii) \textit{Fat Tree topology}: 5 edges, with each host per edge.

For each data point, the source performs RB $2,000$ times with message size $1,024$ bytes. We record the throughput, calculated as the number of reliable-accept per seconds. We have test our Hash- and EC-BRB's and Bracha's RB \cite{Bracha:1987:ABA:36888.36891}. We also test a non-fault-tolerant broadcast (denoted as Broadcast in the table) as a baseline. In Broadcast, the source simply performs $n$ unicasts, and each node accepts a message when it receives anything from source. It provides the highest performance and no fault-tolerance.

The results are presented in Tables \ref{t:linear}, \ref{t:tree}, and \ref{t:fat-tree}.
It is clear that each algorithms perform differently in different topology, but it is difficult to observe a meaningful pattern. This is also why we believe RMB is of interests for practitioners. It provides a lightweight evaluation for different protocols in different scenarios.

One interesting pattern is that with limited bandwidth, our algorithms outperforms Bracha's except for EC-BRB[3f+1] in certain scenarios due to high computation cost. Moreover, H-BRB[3f+1] are within $50\%$ of Broadcast's performance. 

\begin{table}[ht]
    \centering
    \begin{tabular}{|c|c|c|}
    \hline
    Algorithm  & Bandwidth(Mbits/s) & Throughput \\
    \hline
    \hline
    H-BRB[3f+1]  & unlimited & 1041\\
    EC-BRB[3f+1]  & unlimited & 1047\\
    EC-BRB[4f+1] & unlimited & 967\\
    Broadcast  & unlimited & 16303\\
    Bracha  & unlimited & 4604\\
    H-BRB[3f+1] & 42 & 746\\
    EC-BRB[3f+1]  & 42 & 301\\
    EC-BRB[4f+1]  & 42 & 426\\
    Broadcast  & 42 & 1131\\
    Bracha  & 42 & 338\\
    \hline
    \end{tabular}
    \caption{Linear Topology}
    \label{t:linear}
\end{table}

\begin{table}[ht]
    \centering
    \begin{tabular}{|c|c|c|}
    \hline
    Algorithm  & Bandwidth(Mbits/s) & Throughput \\
    \hline
    \hline
    H-BRB[3f+1]  & unlimited & 1117\\
    EC-BRB[3f+1]  & unlimited & 1152\\
    EC-BRB[4f+1] & unlimited & 965\\
    Broadcast  & unlimited & 14604\\
    Bracha  & unlimited & 4521\\

    H-BRB[3f+1]  & 42 & 699\\
    EC-BRB[3f+1]  & 42 & 301\\
    EC-BRB[4f+1]  & 42 & 414\\
    Broadcast  & 42 & 1139\\
    Bracha  & 42 & 334\\
    \hline
    \end{tabular}
    \caption{Tree Topology}
    \label{t:tree}
\end{table}

\begin{table}[ht]
    \centering
    \begin{tabular}{|c|c|c|}
    \hline
    Algorithm  & Bandwidth(Mbits/s) & Throughput \\
    \hline
    \hline
    H-BRB[3f+1]  & unlimited & 1490\\
    EC-BRB[3f+1]  & unlimited & 1072\\
    EC-BRB[4f+1] & unlimited & 772\\
    Broadcast  & unlimited & 16080\\
    Bracha  & unlimited & 4216\\
    H-BRB[3f+1] & 42 & 754\\
    EC-BRB[3f+1]  & 42 & 370\\
    EC-BRB[4f+1]  & 42 & 532\\
    Broadcast  & 42 & 1141\\
    Bracha  & 42 & 518\\
    \hline
    \end{tabular}
    \caption{Fat Tree Topology}
    \label{t:fat-tree}
\end{table}

\subsubsection*{Synchrony vs. Asynchrony}

Even though synchronous Byzantine agreement protocols do not work in asynchrony in general, they serve as a good baseline.\footnote{Some might adapt synchronous algorithms to work in a practical setting. For example, in \cite{Liang2012}, it is argued that the proposed synchronous algorithms are appropriate in a datacenter setting.} In Table \ref{t:sync}, we compare our algorithms with NCBA and Digest from \cite{Liang2012}. NCBA is an erasure-coding based algorithm, whereas Digest is a hash-based algorithm. Both do not work if a node may crash fail in an asynchronous network. In this set of experiment, we adopt the topology of a single switch, $n=4$, message size $1024$ bytes, and $2,000$ reliable broadcasts.

Interestingly, the performances are close. H-BRB[3f+1] even beats the two synchronous algorithms by around $20\%$. Note that in our implementation, we favor NCBA and Digest by skipping the expensive dispute control phase.

\begin{table}[ht]
    \centering
    \begin{tabular}{|c|c|}
    \hline
    Algorithm & Throughput \\
    \hline
    \hline
    H-BRB[3f+1] & 2038\\
    EC-BRB[3f+1] & 1411\\
    EC-BRB[4f+1] & 1295\\
    Digest & 1674\\
    NCBA & 1601\\
    \hline
    \end{tabular}
    \caption{Synchronous vs Asynchronous}
    \label{t:sync}
\end{table}

\subsubsection*{Hash or EC?}
In the final set of experiments, we provide guidance on how to pick the best algorithms given the application scenario. We use the single switch topology, $n=20$, and $100$ rounds of reliable broadcasts. Each experiment has the configuration below:

\begin{itemize}
    \item Exp1: $f = 4$, source's bandwidth limitation $= 50$ KBytes/s, message size $= 1096$ bytes.
    
    \item Exp2: $f = 4$, source's bandwidth limitation $= 500$ KBytes/s, message size $= 1096$ bytes.
    
    \item Exp3: $f = 1$, source's bandwidth limitation $= 50$ KBytes/s, message size $= 1020$ bytes.
    
    \item Exp4: $f = 1$, source's bandwidth limitation $= 500$ KBytes/s, message size $= 1020$ bytes.
\end{itemize}

The result is presented in Table \ref{t:Hash-EC}. The numbers  follow the theoretical analysis: 
(i) H-BRB does not perform well with limited source's bandwidth; (ii) EC-BRB[3f+1] performs better with larger $f$; and (iii) EC-BRB[4f+1] performs better with smaller $f$.

\begin{table}[ht]
    \centering
    \begin{tabular}{|c|c|c|c|}
    \hline
    & H-BRB[3f+1] & EC-BRB[3f+1] & EC-BRB[4f+1]   \\
    \hline
    \hline
    exp1 & 1.6 & 2.5 & 1.2 \\
    exp2 & 19 & 7 & 1.2 \\
    exp3 & 1.6 & 1.3 & 2.3 \\
    exp4 & 19.3 & 12.6 & 16 \\
    \hline
    \end{tabular}
    \caption{Hash vs. EC}
    \label{t:Hash-EC}
\end{table}

%% file: related.tex
\section{Related Work}
\label{s:related}



Reliable broadcast has been studied in the
context of Byzantine failures since the eighties~\cite{Bracha:1987:ABA:36888.36891, BirmanJ87, ChangM84, Garcia-MolinaS91}. 
Bracha~\cite{Bracha:1987:ABA:36888.36891} proposed a clean RB protocol. Birman and Joseph~\cite{BirmanJ87} introduce a multicast and broadcast protocol that enforces causal delivery in the  crash-recovery model. 
Chang and Maxemchuk~\cite{ChangM84} designed a family of reliable broadcast protocols with tradeoffs between the number of helper (low level) messages per high level broadcast message, the local storage requirements, and the fault-tolerance of the system. 
Eugster et al.~\cite{EugsterGHKK03} present a probabilistic gossip-based broadcast algorithm, which is  lightweight and scalable in terms of  throughput and memory management. Guerraoui et al.~\cite{GuerraouiKMPS19} generalize the Byzantine reliable broadcast abstraction to the probabilistic setting, allowing
each of the properties to be violated with a small probability to attain logarithmic per-node communication and
computation complexity. Lou and Wu~\cite{LouW04} introduce double-covered broadcast (DCB), to improve the performance of reliable broadcast in Mobile ad hoc networks (MANETs) by taking advantage of broadcast redundancy to improve the delivery ratio even in systems with high transmission error rate. Pagourtzis et al.~\cite{PagourtzisPS17}  study the Byzantine-tolerant Reliable Broadcast problem  under the locally bounded adversary model and the general adversary model and explore the tradeoff between the level of topology knowledge and the solvability of the problem.
Raynal~\cite{Raynal18} provides a detailed history of how the reliable broadcast abstraction has evolved since the eighties. He also illustrates the importance of reliable broadcast in systems with crash and Byzantine failures. A survey of existing reliable broadcast mechanism can be found in ~\cite{DefagoSU04,BaldoniCM06}. Bonomi et al.~\cite{BonomiFT2019} implement Byzantine-tolerant reliable broadcast in multi-hop networks connected by cannels that are not authenticate and improve previous bounds while preserving both the safety and liveness properties of the original definitions~\cite{Dolev81}.

In recent years there have been several attempts at implementing reliable broadcast more efficiently~\cite{LiangV11, PatraR11,Patra11, Liang2012,Drabkin, JeanneauRAD17}.
Jeanneau et al. ~\cite{JeanneauRAD17} perform crash-tolerant  reliable broadcast in a wait-free system ($f\le n-1$) equipped with a failure detector. Messages are transmitted over spanning trees that are built dynamically on top of a logical hierarchical hypercube-like topology built based on failure detectors providing logarithmic guarantees on latency and number of messages disseminated. 
Liang et al.~\cite{Liang2012} compare the performance of several synchronous Byzantine agreement algorithms NCBA and Digest. Fitzi and Hirt~\cite{FitziH06} present a synchronous RB protocol where the message is broadcast only $n$ times compared to prior implementations that use $n^2$ broadcasts. Patra~\cite{Patra11} presents an algorithm for Byzantine RB with optimal resilience and bit complexity $O(nL)$. In~\cite{PatraR11}, Patra and Rangan reported similar results that were only problematically correct. Choudhury \cite{Choudhury17} achieves Byzantine-tolerant reliable broadcast with just majority correctness whilst maintaining optimal bit complexity of $O(nL)$, using a non-equivocation mechanism provided by hardware. 

\section{Conclusion}

In this paper we present a family of reliable broadcast algorithms that tolerate faults ranging from crashes to Byzantine ones that are suited to bandwidth constrained networks using cryptographic hash functions and erasure codes. We provide theoretical algorithms, experimental results and impossibility proofs to state our case.
We hope this paper will provide guidance and reference for practitioners that work on fault-tolerant distributed systems and use reliable broadcast as a primitive. 

\section*{Acknowledgements}
This research is supported in part by National Science Foundation award CNS1816487. Any opinions, findings, and conclusions or recommendations expressed here are those of the authors and do not necessarily reflect the views of the funding agencies or the U.S. government.